\documentclass{llncs}
\usepackage{llncsdoc}

\usepackage{amssymb}
\usepackage{graphicx}
\usepackage{amssymb}
\usepackage[mathcal]{euscript}
\usepackage{amssymb,amsmath}
\usepackage{graphicx}
\usepackage{tikz}
\usetikzlibrary{arrows,backgrounds,snakes}
\usepackage{amsmath,amstext,amsfonts,amssymb}
\usepackage{amsbsy}
\usepackage{multirow}
\usepackage{color}

\definecolor{darkred}{rgb}{0.5,0.0,0.0}
\definecolor{darkblue}{rgb}{0.2,0.1,0.6}
\newtheorem{dfn}{Definition}

\usepackage{url}
\urldef{\mailsa}\path|{alfred.hofmann, ursula.barth, ingrid.haas, frank.holzwarth,|
\urldef{\mailsb}\path|anna.kramer, leonie.kunz, christine.reiss, nicole.sator,|
\urldef{\mailsc}\path|erika.siebert-cole, peter.strasser, lncs}@springer.com|
\newcommand{\keywords}[1]{\par\addvspace\baselineskip
\noindent\keywordname\enspace\ignorespaces#1}

\begin{document}
\pagestyle{empty}
\mainmatter  

\title{Studying Node Cooperation in Reputation Based Packet Forwarding within Mobile Ad hoc Networks}

\titlerunning{Studying Cooperation in Reputation Based Packet Forwarding in MANETs}

%
%
\author{Sara Berri\inst{1,2}\and Vineeth Varma\inst{3}\and Samson Lasaulce\inst{2}\and Mohammed Said Radjef\inst{1}\and Jamal Daafouz\inst{3}\thanks{The present work is supported by the LIA project between CRAN, Lorea and the International University of Rabat. }}
\authorrunning{S. Berri\and V. Varma\and S. Lasaulce\and M.S. Radjef\and J. Daafouz }

\institute{Research Unit LaMOS (Modeling and Optimization of Systems), Faculty of Exact Sciences, University of Bejaia, Bejaia, 06000, Algeria
\and L2S (CNRS-CentraleSupelec-Univ. Paris Sud), Gif-sur-Yvette, France \and CNRS and Universit\'e de Lorraine, CRAN, UMR 7039, France\\ \{sara.berri, samson.lasaulce\}@l2s.centralesupelec.fr, \{vineeth.satheeskumar-varma, Jamal.Daafouz\}@univ-lorraine.fr, radjefms@gmail.com}

%
%

\toctitle{Studying Node Cooperation in Reputation Based Packet Forwarding within MANETs}
\tocauthor{}
\maketitle

\begin{abstract}
In the paradigm of mobile Ad hoc networks (MANET), forwarding packets originating from other nodes requires cooperation among nodes. However, as each node may not want to waste its energy, cooperative behavior can not be guaranteed. Therefore, it is necessary to implement some mechanism to avoid selfish behavior and to promote cooperation. In this paper, we propose a simple quid pro quo based reputation system, i.e., nodes that forward gain reputation, but lose more reputation if they do not forward packets from cooperative users (determined based on reputation), and lose less reputation when they chose to not forward packets from non-cooperative users. Under this framework, we model the behavior of users as an evolutionary game and provide conditions that result in cooperative behavior by studying the evolutionary stable states of the proposed game. Numerical analysis is provided to study the resulting equilibria and to illustrate how the proposed model performs compared to traditional models.
\end{abstract}

\section{Introduction}\label{sec:Intro}
A mobile ad hoc network (MANET) is a wireless multi-hop network formed by a set of mobile independent nodes. A key feature about MANETs is that they are self organizing and are without any established infrastructure. The absence of infrastructure implies that all networking functions, such as packet forwarding, must be performed by the nodes themselves \cite{book1}. Thus, multi-hop communications rely on mutual cooperation among network's nodes. As the nodes of an ad hoc network have limited energy, the nodes may not want to waste their energy by forwarding packets from other nodes. If all the nodes are controlled by a central entity, this will not be a major issue as cooperation can be a part of the design, but in applications where each node corresponds to an individual user, it is crucial to develop mechanisms that promote cooperation among the nodes.

Several works in the literature provide solutions based on incentive mechanisms, such as those based on a credit concept \cite{c2i8}, \cite{c2i9}, \cite{ref2016} etc., whose idea being that nodes pay for using some service and they are remunerated when they provide some service (like packet forwarding). Others like \cite{a4}, \cite{a5} use reputation-based mechanisms to promote cooperation. Game theory has been a vital tool in literature to study the behavior of self-serving individuals in serval domains including MANETs. In \cite{hub}, \cite{b1}, \cite{c2i7}, etc. the interaction among nodes in packet forwarding is modeled as a one shot game based on prison's dilemma model, extended then to repeated game. Furthermore, evolutionary game theory is introduced in \cite{ser1}, \cite{ser}, \cite{evo}, to study the dynamic evolution of system composed of nodes and to analyze how cooperation can be ensured in a natural manner. In \cite{ser1} the evolutionary game theory is applied to study cooperation in packet forwarding in mobile ad hoc networks. Here, the authors used the prison's dilemma-based model \cite{hub} and the aim was to implement several strategies in the game and to evaluate performance, by observing their evolution over time.

The aforementioned works rely on incentive mechanisms, which has been proved to improve nodes cooperation. However, implementing such solutions often result in a large computational complexity during the game. We would like to find an answer to the following question, "Is it possible to achieve global cooperation in packet forwarding by a simple and natural way?". In this paper, we model the nodes interaction in a MANET as an evolutionary game by proposing a new formulation of the packet forwarding and reputation model. We introduce a simple reputation system with a quid pro quo basis, wherein, reputation is gained by forwarding packets and is lost when refusing to forward. However, a key feature is that the reputation loss depends on the packet source. If the packet is from a node with low reputation, less reputation is lost by not forwarding that packet. This simply means that selfish users will naturally have low reputation, while users are encouraged to help other cooperative users, resulting in a significantly different model from the likes of \cite{ser1}, \cite{evo} etc. With this model, we study two node classes, one which try to maintain a certain high reputation, and another class which disregard their reputation. We show that nodes are likely to cooperating by means of evolutionary game theory concepts and provide numerical results showing how the proposed model improves network performance. The novel reputation model we propose will naturally result in the cooperative users cooperating among each other and refusing to forward packets from selfish users, thereby eliminating the need for a third party to punish selfish behavior.

The remainder of the paper is structured as follows. In Sec. \ref{sec:Gamemod} we formulate the reputation and game models. We propose to analyze the evolutionary game in Sec. \ref{sec:evolanal}, by providing the associated equilibrium, and studying strategies evolution. This allows us to determine condition ensuring global cooperative behavior. The numerical results are presented in Sec. \ref{sec:numerical analysis}; it has to be noted that the results are the same for any game settings satisfying the provided conditions and not only for the given examples. Finally, Sec. \ref{sec:conc} presents the conclusion.




\section{Problem formulation and proposed game model}\label{sec:Gamemod}
 In this section, we provide a game model to study the packet forwarding interaction. We consider a packet forwarding game, where the players are the nodes, each of them can be cooperative, by forwarding other nodes' packets, or non-cooperative, by dropping other nodes' packets. Thus, the players have to choose a strategy $s_i$ from the strategy set $\mathcal{S}$=\{C, NC\}. The actions C and NC mean cooperative and non-cooperative, respectively. The two player packet forwarding game can be defined in its strategic-form as following.
  \begin{equation}\label{eqgame}
\mathcal{ G}^{(2)}= <\{1, 2\},~\{\mathcal{S}_{i}\}_{i\in \mathcal{I}},~\{u_{i}\}_{i\in \mathcal{I}} >,
  \end{equation}
  where:
  \begin{description}
    \item[$\bullet$] $\mathcal{I}=\{1,2\}$ is the set of players (two players), that are the network nodes;
    \item[$\bullet$] $\mathcal{S}_i$ is the set of pure strategies of player $i \in \mathcal{I}$, which is the same for all the players, corresponding to $\mathcal{S}$=\{C, NC\};
\item [$\bullet$] $u_i$ is the utility of player $i \in \mathcal{I}$, that depends on its behavior and that of its opponent. To demonstrate the utility formulation, we consider the case of a pair of nodes from the network, within which a node may act as a sender and a relay (and vice versa).  Thus, the players' utility can be represented by a payoff matrix as given by  \eqref{payoffmatrix}.
\begin{equation}\label{payoffmatrix}
    A=\bordermatrix{
&\mathrm{C}&\mathrm{NC}\cr \mathrm{C}& \lambda-1&  -1 \cr \mathrm{NC}& \lambda& 0\cr },
   \end{equation}
  \end{description}
where: $\lambda>0$ is a coefficient representing the benefit associated to successfully sending a packet while spending a unit of energy. The first player actions are along the rows and the second players along the columns. Naturally, when $\lambda<1$ no nodes are motivated to cooperate as the energy cost relative to the gain from having packets relayed is too high. In the interesting case (the case where a MANET framework is feasible) of $\lambda>1$, the outcome of the proposed game can be characterized by the well-known Nash Equilibrium (NE), which is the strategy profile from which no player has interest in changing unilaterally its strategy. The resulting strategy profile is beneficial for players when they act individually. However, the NE of the packet forwarding game is inefficient, corresponding to drop all the time, and provides for players $0$ as utility. Thus, to overcome this problem we propose to add to the game \eqref{eqgame} a reputation model, that defines the reward and the cost in terms of reputation according to the node decision, cooperative or non-cooperative. On the other hand, it would be better to model the interactions among all the $N$ nodes and not just the two-player case. To deal with this, we propose to introduce evolutionary game theory, where the dynamical evolution of game strategies is studied through pairwise interactions.

In the following section, we provide the reputation model we propose, and construct the new packet forwarding game including the reputation mechanism as an integrated system. That means the game is played taking into account the reputation, which we show can be interpreted as a constraint on the strategy space, while the nodes aim to maximize their utility function.
 \subsection{Reputation model}

We assume that there is a reputation system introduced in order to discourage selfish behavior and reward cooperative behavior by separating these two classes of nodes. The reputation system is represented as a function depending on the own action and the opponent's action. The reputation increases by a certain margin $\delta_r$ whenever a node relays the packet from another node, chooses C as action. Reputation is lost whenever a node refuses to relay a packet, by choosing NC as action. However, the loss of reputation from refusing to relay the packet from a node with low reputation $\delta_b$ is smaller than the loss incurred by refusing to relay the packet from a well reputed node $\delta_g$. For ease of notation the reputation of a user $i \in \mathcal{I}$ is given by $R_i(t)$, and if $R_i(t) >0$ the node has a good reputation and otherwise bad reputation. The change in reputation is given by:
\begin{equation}
R_i(t+1)= R_i(t) + d_i(t) \delta_r - (1-d_i(t)) ( \delta_g \mathbf{1}(R_{j}>0) +     \delta_b \mathbf{1}(R_{j}\leq 0) ),
\end{equation}
 where: $d_i(t) \in \{0,1\}$ is the decision to relay or not at time $t$. 0 corresponds to the action NC, and 1 to C. $j$ is a random variable indicating the sender requesting $i$ to relay. $\mathbf{1}$ is the indicator function, it is one when the condition inside the brackets is satisfied and $0$ otherwise.

We consider two primary classes of nodes based on their reputation value. The set $H$ of "Hawks" who are selfish (non-cooperative) and don't care about reputation. As a result these nodes never relay packets, i.e., $s_i=$NC $\forall i \in$ $H$, but use the network and try to make the other nodes relay their packets, and so we have $R_i(t) <0$ $\forall i \in H$. These nodes will always have $d_i(t)=0$ for all $t$ and therefore will also have a low reputation.

The other class of nodes are the set $D$ of "Doves", who try to maintain a positive reputation. These nodes will have a strategy $s$ such that on average their reputation gain is positive. Let us denote the dove population share (fraction of users who are in the dove class) by $p$. The population share of hawks will simply be given by $1-p$.

\subsection{Utility maximization}
In this subsection, we present how reputation system is integrated in the packet forwarding game in order to improve game outcomes and avoiding the non-cooperative situation. As even the doves do not want to waste energy, they will not attempt to transmit ever single packet, but only such that their average reputation gain is at least $0$ (reputation must be an increasing function). We assume that even a cooperative node, i.e., the Dove class, does not relay packets all the time. The doves have a mixed strategy to relay messages, and it relays messages from other doves with probability $s_d$, and from hawks with a probability $s_h$, i.e., the action C is chosen with different probabilities depending on the opponent's class.

As a result, the net utility is given by the number of times their packets get forward subtracted by the energy cost paid is given by them. The expected payoff of doves is given by the formula \eqref{eqUtiD}.
\begin{equation}\label{eqUtiD}
U(D,p) =  (\lambda-1) p s_d  - s_h (1-p).
\end{equation}
This must be maximized over $s_d,s_h$ while maintaining a positive reputation, i.e., $\mathbb{E}[R_i(t+1)-R_i(t)] \geq 0$ or
\begin{equation}\label{eq:repcon}
p( s_d \delta_r - (1-s_d) \delta_g ) + (1-p) ( s_h \delta_r - (1-s_h) \delta_b) \geq 0.
\end{equation}

Therefore for a given population share of doves $p$, we can find the strategy of doves by solving the following optimization problem.
$$\max_{s_d,s_h} U(D,p)$$
\begin{equation}\label{Opt1}
\begin{array}{ll}
  p( s_d \delta_r - (1-s_d) \delta_g ) + (1-p) ( s_h \delta_r - (1-s_h) \delta_b) \geq 0 \\
 0 \leq s_d\leq 1,~ 0 \leq s_h\leq 1.  \\
     \end{array}
     \end{equation}

Hawks have the same utility function, but don't have a reputation constraint, therefore, the corresponding expected payoff is given by the formula \eqref{eqUtiH}.
\begin{equation}\label{eqUtiH}
U(H,p)= \lambda p s_h.
\end{equation}
Thus, the expected payoff of any individual is given by \eqref{eqEXp}:
 \begin{equation}\label{eqEXp}
    U(p,p)=pU(D,p)+(1-p)U(H,p),
\end{equation}
where $p$ is the population profile.

If $\lambda>1$, $U(D,p)$ is maximized trivially by choosing $s_d=1$. Therefore, $s_h$ will be the smallest such that constraint (\ref{eq:repcon}) holds.
\begin{equation}
\begin{array}{rll}
 \nonumber
 p( s_d \delta_r - (1-s_d) \delta_g ) + (1-p) ( s_h \delta_r - (1-s_h) \delta_b) & \geq & 0 \\ \nonumber
\Rightarrow p\delta_r + (1-p) ( s_h \delta_r - (1-s_h) \delta_b)&\geq&0  \\ \Rightarrow \nonumber \label{eq1proof}
    p\delta_r + (1-p) (s_h (\delta_r+\delta_b)-\delta_b)&\geq&0  \\\Rightarrow  \nonumber
s_h\geq  \frac{ (1-p) \delta_b-p \delta_r } { (\delta_b + \delta_r)(1-p) }.
\end{array}
\end{equation}
Thus, we have:
\begin{equation}\label{eqq}
    s_h^{\star}= \max\left\{\frac{ (1-p) \delta_b-p \delta_r } { (\delta_b + \delta_r)(1-p) },0 \right\}.
\end{equation}

Note that the introduction of $s_h<1$ is one of the main novelties of this paper, which can be different from $s_h=1$ as defined in the traditional forwarding game payoff  \cite{hub}, \cite{b1}, \cite{c2i7}, \cite{ser1}, \cite{evo}, etc. where the cooperative nodes forward packet all the time without making distinction among the opponent nodes that can be cooperative, belonging to $D$, or non-cooperative, belonging to $H$. Furthermore, the proposed reputation system is simpler and can be defined as a constraint when nodes take decision purely based on the reputation class of the packet source node.

\subsection{Evolutionary game formulation}

We can formally define the resulting evolutionary game with the strategic form
\begin{equation}\label{eqevogame}
\mathcal{G}= < \{ D, H \}, \{ (s_d,s_h )\} \times \{0\}, p \in [0,1],   \{u_{c}\}_{c\in \{D,H\} } >,
  \end{equation}
  where:
  \begin{description}
    \item[$\bullet$] $\{D,H\}$ are the reputation classes (or population types);
     \item[$\bullet$] $\{ (s_d,s_h )\}$ is the set of strategies playable by $D$, with $H$ always playing $0$ or NC strategy;
    \item[$\bullet$] $p$ is the population share of class $D$;
    \item[$\bullet$] $u_{c}$ is the utility of class $D$ or $H$ as defined in (\ref{eqUtiD}) and (\ref{eqUtiH}).
\end{description}

Our objective in the following section is to study the evolution of strategies in this game, and analyze possible equilibrium points.

\section{Evolutionary game analysis}\label{sec:evolanal}
Evolutionary game theory study the dynamic evolution of a given population based on two main concepts: evolutionary stable strategy (ESS) and replicator dynamics. Let $p$ the initial population profile. We assume that a proportion $\varepsilon$ of this population plays according to another profile $q$ (population of mutants), while the other individuals keep their initial behavior $p$. Thus, the new population profile is $(1-\varepsilon)p+\varepsilon q$.  The expected payoff of a player that plays according to $p$  is $U(p,(1-\varepsilon)p+\varepsilon q)$, and it is equal to $U(q,(1-\varepsilon)p+ \varepsilon q)$ for the one playing according to $q$.
 \begin{dfn} \emph{\cite{Maynard73}}
 A strategy $p \in \Delta$ is an evolutionary stable strategy \emph{(ESS)}, if :
$\forall q \in \Delta, \exists~ \bar{\varepsilon} = \bar{\varepsilon}(q) \in (0,1), \forall \varepsilon \in (0,\bar{\varepsilon})$
\begin{equation}\label{eq2}
    U(p, (1-\varepsilon) p + \varepsilon q) > U(q, (1-\varepsilon) p + \varepsilon q),
 \end{equation}
  $\bar{\varepsilon}$ is called \emph{invasion barrier} of the strategy $p$, which may depend on $q$.  
\label{def1}
\end{dfn}
The replicator dynamics is the process that specifies how a population is distributed over the pure strategies set in a game evolving in time.
\begin{dfn}\label{def2} \emph{(Replicator dynamics)}.
The replicator dynamics is given by (\ref{eq3}) \cite{l5}:
 \begin{equation}\label{eq3}
\dot{p_{i}}=p_{i}[U(i,p)-\sum_{j=1}^{|\mathcal{S}|}p_jU(j,p)], ~~~~~i\in \{1,\ldots,|\mathcal{S}|\}.
\end{equation}
\end{dfn}
The system (\ref{eq3}) describes the replication process in continuous time. It gives the percentage of individuals newly playing strategy $s_i$ in the next period, it depends on the initial value $p_i (t_0)$.
Using the relation \eqref{eq3}, the replicator dynamics of the proposed game is:
\begin{equation}\label{eqREP}
 \dot{p} = p(1-p)(p(\lambda-1)(1-s_h^{\star})-s_h^{\star}).
 \end{equation}
For the evolutionary game $\mathcal{G}$, we have the following results.
\begin{theorem}\label{theo1}
When $\lambda>1$, the evolutionary game $\mathcal{G}$ admits exactly two ESS at $p^{\star}=0$ with an invasion barrier $\bar{\varepsilon} = \min\left\{\frac{\delta_b}{\delta_r q(\lambda-1)},1\right\}$, and $p^{\star}=1$ with an invasion barrier $\bar{\varepsilon} =1$. When the initial configuration is such that $p<p_T$, the replicator dynamics takes the system to $p^{\star}=0$ and when $p>p_T$, the replicator dynamics takes the system to $p^{\star}=1$ with
\[p_T = \frac{\delta_b}{\delta_b+\lambda \delta_r} \]
corresponding to the mixed NE.
\end{theorem}

\begin{proof}
First, we can easily verify that $p_T$ corresponds to a mixed NE by noticing that the utilities of $H$ and $D$ classes are identical at this point. Next, we use the definition \ref{def1}, to prove the results stated in Theorem \ref{theo1} corresponding to the invasion barrier. Let $x=(1-\varepsilon) p + \varepsilon q$, and $\bar{U}=U(p, x)-U(q, x)$. $U(p, x)$ and $U(q, x)$ are defined using the relation \eqref{eqEXp}.
\begin{eqnarray}
 \nonumber
 \bar{U}&=&  p(x(\lambda-1)-s_h^{\star}(1-x))+(1-p)(\lambda x s_h^{\star})-q(x(\lambda-1)-s_h^{\star}(1-x))-(1-q)(\lambda x s_h^{\star})\\ \nonumber \label{eq1proof}
   &=& (p-q)(x(\lambda-1)-s_h^{\star}(1-x))-(p-q)(\lambda x s_h^{\star})\\ \nonumber
   &=&(p-q)(x(\lambda-1)(1-s_h^{\star})- s_h^{\star}).  \\ 
\end{eqnarray}
\begin{enumerate}

\item In the first case, $p^{\star}=0$, this gives $s_h^{\star}=\frac{\delta_b}{\delta_b+\delta_r}$. We can solve for the condition when
  \begin{eqnarray}
 \nonumber
\bar{U}&>& 0\\ \nonumber
\Rightarrow -q(\varepsilon q(\lambda-1)(1-s_h^{\star})- s_h^{\star})&>& 0 \\ \nonumber \Rightarrow
(-\varepsilon q(\lambda-1)(1-s_h^{\star})+ s_h^{\star})&>&0 \\ \nonumber \label{eq2proof}
   \Rightarrow (-\varepsilon q(\lambda-1)\delta_r+ \delta_b)& > &0\\ \Rightarrow  \varepsilon < \frac{\delta_b}{\delta_r q(\lambda-1)}.\\ \nonumber
    \end{eqnarray}
Thus, from the definition \ref{def1} we conclude that $p^\star=0$ is an ESS with $\bar{\varepsilon}=\min\left\{\frac{\delta_b}{\delta_r q(\lambda-1)},1\right\}$ as an invasion barrier. Note that $p^{\star}=0$ is an ESS only if the population share of $D$ decreases, and that of $H$ increases, i.e., the replicator dynamics is negative or $\dot{p_{i}}<0$. If $s_h=s^{\star}_h$, $\dot{p}<0$ gives the following result:
 \begin{eqnarray}\label{eq3proof}
 \nonumber
 \dot{p}&<& 0 \\  \nonumber \Rightarrow
  p(1-p)(p(\lambda-1) -\frac{ (1-p) \delta_b-p \delta_r } { (\delta_b + \delta_r)(1-p)}(p(\lambda-1)+1))&<&0 \\  \nonumber \Rightarrow
 \frac{p((\delta_b+\delta_r \lambda)p-\delta_b)}{\delta_b+\delta_r}&<&0 \\  \nonumber \Rightarrow
  p&<&p_T\\ 
\end{eqnarray}
\item Now we prove that $p^{\star}=1$ is an ESS. In this case $s_h^{\star}=0$. Thus, the following result:
  $$\bar{U}=(1-\varepsilon(1-q))(\lambda-1),$$
  we have $\lambda>1$ $\Rightarrow$ $\bar{U}>0$. This implies, according to definition \ref{def1}, that $p^{\star}=1$ is an ESS $\bar{\varepsilon}=1$ as an invasion barrier. This  occurs if $\dot{p}>0$, i.e., when:
  \begin{equation}\label{rep}
    \frac{p((\delta_b+\delta_r \lambda)p-\delta_b)}{\delta_b+\delta_r}>0\Rightarrow p>\frac{\delta_b}{\delta_b+\delta_r\lambda}=p_T
  \end{equation}

\end{enumerate}

 \begin{flushright}
  $\square$
  \end{flushright}
\end{proof}

\section{Numerical analysis}\label{sec:numerical analysis}

In this section, we present numerical application of the proposed evolutionary game including a reputation system. All the results are based on the replicator dynamics which describes how the population evolves, and allows one to determine others performance metrics such as expected utility of players and the number of forwarded packets.

\begin{figure}[h!]
\begin{center}
\includegraphics[width=8.5cm,height=6cm]{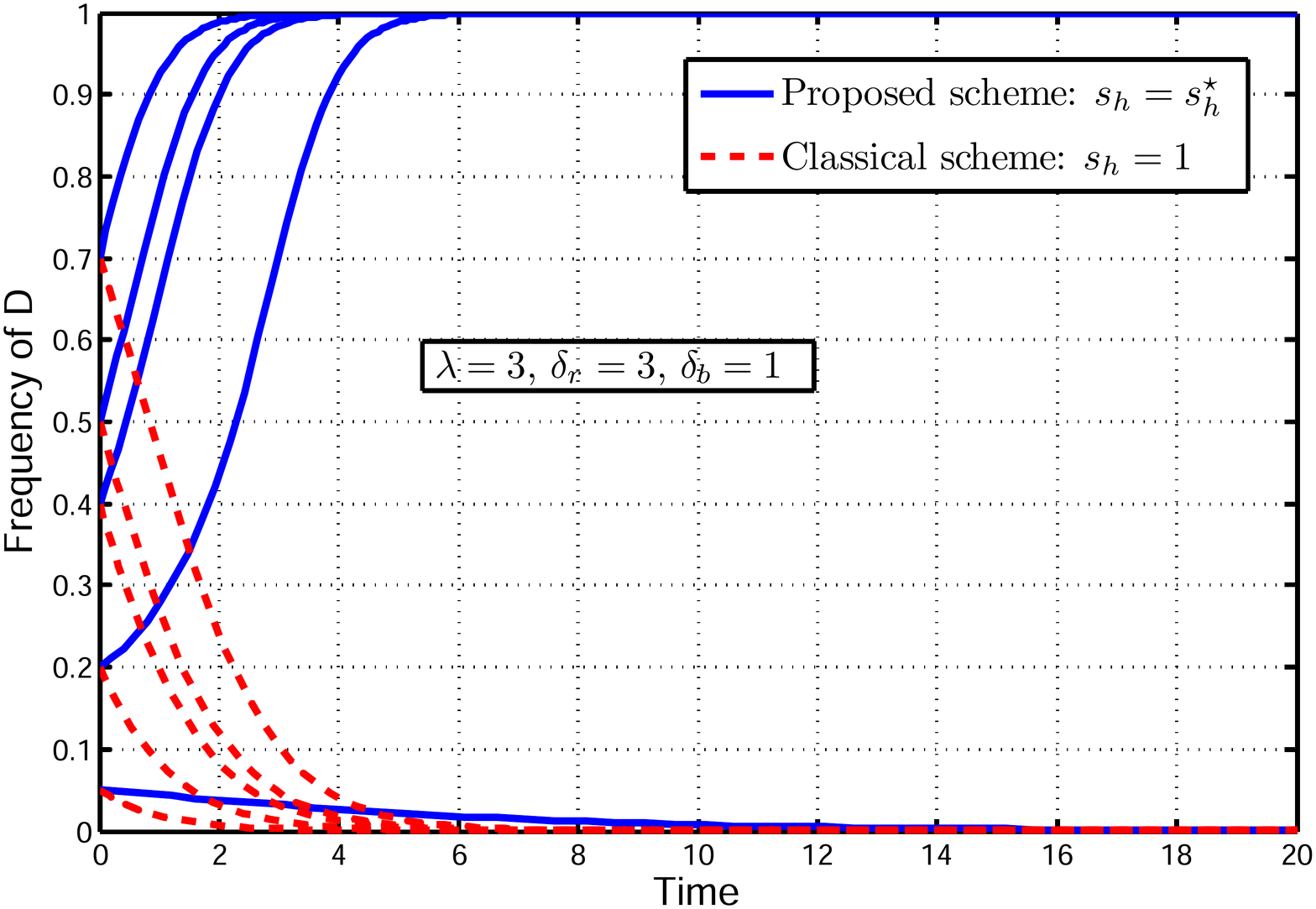}\\
\end{center}
\caption{Evolutionary dynamics of the Doves, nodes that play the strategy 'Cooperation' with a proability $s_h^*$ in the proposed game model, and previous packet forwarding game model where $s_h=1$. We plot for several initial frequency values.}\label{fig2}
\end{figure}

\begin{figure}[h!]
\begin{center}
\includegraphics[width=8.5cm,height=6cm]{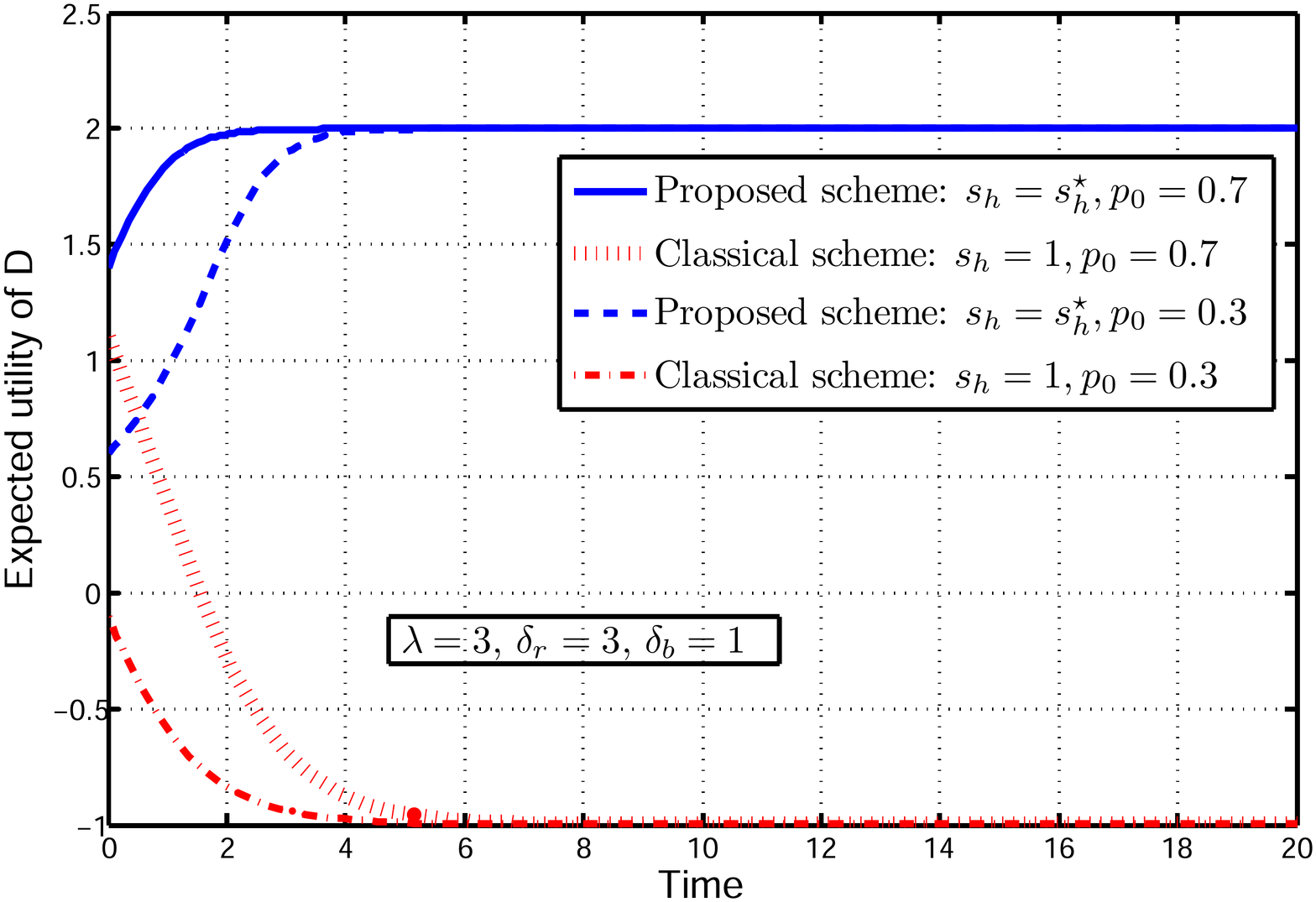}\\
\end{center}
\caption{The expected utility of $D$ in the proposed game model and previous packet forwarding game model, where $s_h=1$. We plot for several initial frequency values 0.7 and  0.3.}\label{fig4}
\end{figure}

\begin{figure}[h!]
\begin{center}
\includegraphics[width=8.5cm,height=6cm]{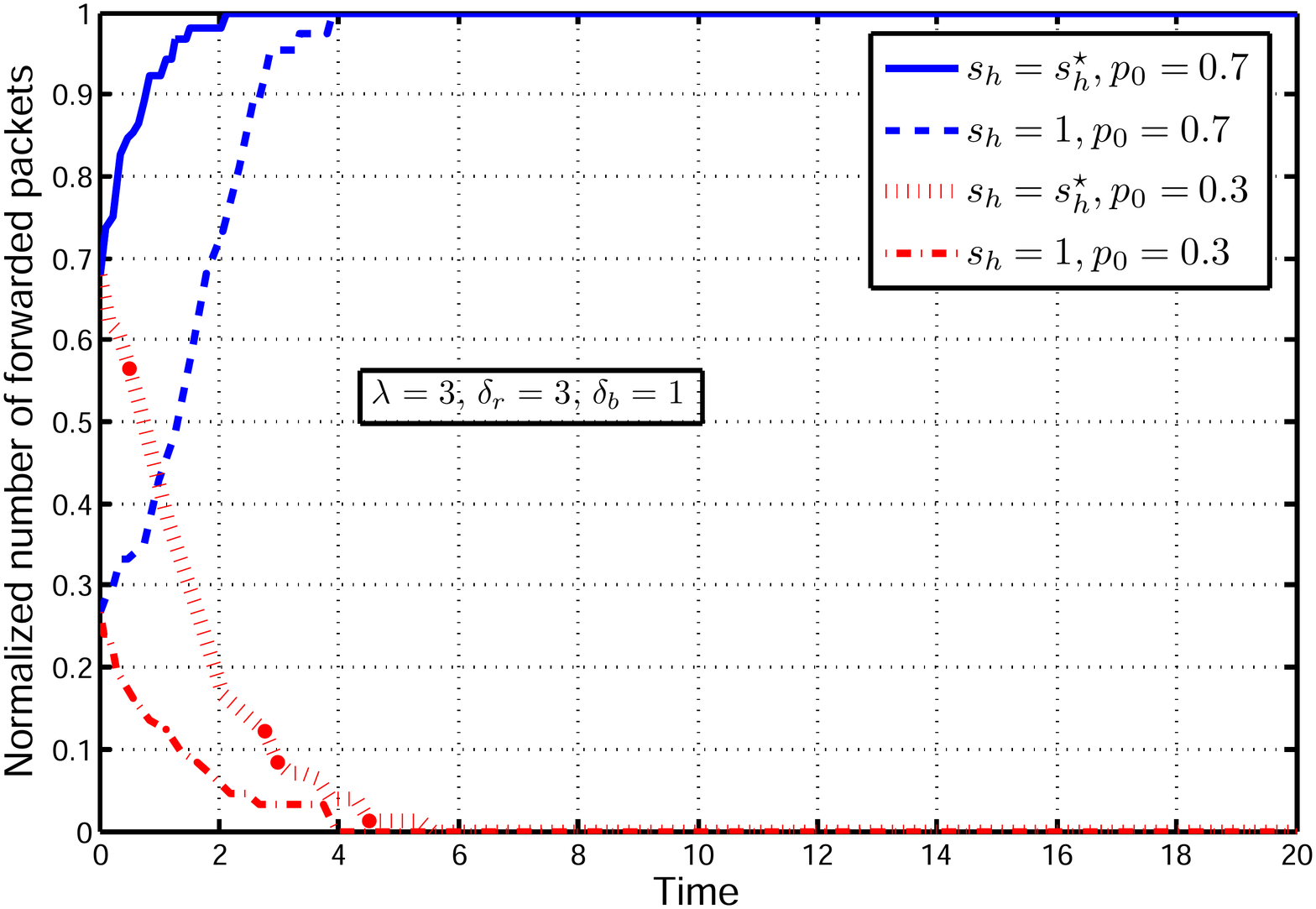}\\
\end{center}
\caption{Normalized number of forwarded packets for an Ad hoc network of $50$ nodes, in the proposed game model and previous packet forwarding game model, where $s_h=1$. We plot for two initial frequency values 0.7 and  0.3.}\label{fig5}
\end{figure}

We study the effect of the proposed reputation model on the evolutionary stable strategy of the game. Fig. \ref{fig2} presents the results, we consider two scenarios: 1) The curve in solid line corresponds to results provided by the proposed game model including a reputation system, and assuming that the cooperative nodes forward packets of non-cooperative nodes with some probability $s_h^{\star}$. 2) The curve in dashed line corresponds to results provided by putting $s_h=1$, meaning that cooperative nodes forward all the time, which corresponds to the previous packet forwarding game introduced in \cite{hub}, \cite{b1}, \cite{c2i7}, \cite{ser1}, etc. It is seen that using our new formulation, which integrates a reputation mechanism as a constraint, the system could converge towards a cooperative state, by carefully choosing the game settings. Thus, global cooperation could be guaranteed after a given time. Whereas, when the game does not include the reputation constraint, the population converges to the strategy non-cooperation, which is the unique evolutionary stable strategy of the game, regardless of the initial condition and game settings.

The results given by the Figure \ref{fig2} can be used to characterize the expected utility of players. Fig. \ref{fig4} presents the results of both cases $s_h=s_h^{\star}$ and $s_h=1$. From these figures, we observe that the utilities evolve over time in the same way that the proportion of the considered population, corresponding to $D$ in that case. Thus, the proposed model provides better results, it promotes cooperation among nodes.

Indeed, in order to show the influence level of these results on network performance, we consider a network composed of $50$ nodes, randomly placed in surface of $1000m \times 1000m$, with a transmission range equals to $150m$, and plot normalized number of forwarded packets within a network, using the proposed game model with constraint, and that introduced in previous works \cite{hub}, \cite{b1}, \cite{c2i7}, \cite{ser1}, etc. defined without any constraint and assuming that the cooperative nodes forward all the time, i.e.,  $s_h=1$. We assume that all the nodes need to send $10$ packets to a given destination. Fig. \ref{fig5} represents the results for the following game settings: $\lambda=3$, $\delta_r=3$ and $\delta_b=1$. It clearly shows a direct influence, because the number of forwarded packets is strongly linked to the cooperative nodes proportion in packet forwarding.

\textbf{Remark: }While setting $\delta_b=0$ can indeed make $p^{\star}=1$ the only ESS, this may not be a suitable reputation model for the MANET framework due to several reasons. Firstly, setting $\delta_b=0$ will completely discourage $D$ from forwarding packets from $H$ class, which may also include new users to the MANET, thereby discouraging new users as they might be unable to send their packets without increasing their reputation. Secondly, note that $D$ may not always forward packets from $D$ as in practice the channel conditions between the nodes also play a big role in determining the resource cost and therefore the utility gained by forwarding (which we have not accounted for in this work). Accounting for channel fading due to path loss or small-scale fading will therefore, be a relevant extension of this work. These considerations show that reputation model parameters must be carefully designed in practice.

\section{Conclusion}\label{sec:conc}
The contribution of this paper is to propose a new formulation of the packet forwarding game \cite{hub}, introducing a reputation system, which modifies reputation based on the reputation class of the packet source, i.e., cooperative or non-cooperative. The aim is to motivate node cooperation using a simple and efficient mechanism. As a smaller reputation is lost by not forwarding packets from selfish users (classified by the reputation system), cooperative users will effectively forward the packets from other cooperative users and may avoid forwarding packets from selfish users. Effectively, we have demonstrated using evolutionary game theory concepts that, global cooperation in the network can be achieved under some conditions we stated related to the game settings with a low computational complexity. Finally, through simulations, we have shown that in terms of the number of forwarded packets in the MANET, the proposed game model provides significant gains over the game model where the cooperative nodes forward packet regardless of the opponent's behavior. As an extension of the present work, we propose to study the multi-hop case, where the interaction involves more than two players. Another relevant extension would be to account for channel fluctuations and the resulting utility function which might result in the cooperative class users forwarding packets from selfish users and not another cooperative user despite the reputation losses, due the channel conditions being favorable.

\end{document}